\documentclass[a4paper,10pt]{article}

\usepackage{amsmath, amssymb, amsthm, amsfonts, amscd, epsfig, graphicx, color}

\newtheorem{theo}{Theorem}[section]
\newtheorem{defi}[theo]{Definition}
\newtheorem{lem}[theo]{Lemma}
\newtheorem{prop}[theo]{Proposition}

\newtheorem{rem}[theo]{Remark}

\newcommand{\R}{\mathbb{R}}

\newcommand{\Z}{\mathbb{Z}}

\newcommand{\D}{\mathcal{D}}
\newcommand{\F}{\mathcal{F}}

\newcommand{\AV}{\mathcal{J}}
\newcommand{\An}{\textnormal{Angv}}
\newcommand{\Corr}{\mathcal{C}}

\newcommand{\m}{\textnormal{min}}
\newcommand{\M}{\textnormal{max}}

\title{How Uncertainty bounds the shape index of simple cells}
\author{D. Barbieri, G. Citti, A. Sarti}

\begin{document}

\maketitle

\begin{abstract}
We propose a theoretical motivation to quantify actual physiological features, such as the shape index distributions measured by Jones and Palmer in cats and by Ringach in macaque monkeys. We will adopt the Uncertainty Principle associated to the task of detection of position and orientation as the main tool to provide quantitative bounds on the family of simple cells concretely implemented in primary visual cortex.
\footnote{
\thanks{The research of the first author was supported by Grant DIM2011 - R\'egion \^{I}le de France.}

\textit{2010 Mathematics Subject Classification}: 62P10, 43A32, 81R15

{\textit{Keywords}: Visual Cortex, Uncertainty Principle, Lie Groups, Receptive Profiles }}

\end{abstract}


\section{Introduction}

One of the fundamental tasks performed by simple cells in primary visual cortex is that of detecting position and local orientation of a stimulus \cite{HWFerrier}. On the other hand, the functional behavior of simple cells as visual detectors is characterized in terms of standard linear filtering and with other so-called non classical behaviors \cite{G}. We will concentrate on linear aspects, and consider classical receptive profiles modeled with a planar oscillation under a spatially localizing window. In \cite{Ringach2002}, such receptive profiles were studied in terms of two dimensionless indexes of shape $(n_x,n_y)$ corresponding to the product of the frequency of the oscillation and the sizes of the window in the direction of the oscillation and in the orthogonal one, showing that the distribution of such feature on V1 simple cells of macaque monkeys is confined to a specific region. This result is summarized in Figure \ref{fig:Ring}. Remarkably, the same confinement was found also in cats \cite{JP1, JP2}, and this suggests that this pattern can be associated with some criteria of optimality with respect to perceptive tasks. Notable proposals of such criteria were stated in terms of sparse coding in \cite{OlshausenField}, already discussed in \cite{Ringach2002}, and more recently in \cite{SSBSL}, or in terms of Bayesian learning \cite{H}.

In this paper we will focus on the task of position and orientation detection, and propose theoretical motivations based on the Uncertainty Principle for the corresponding geometry to explain such confinement. In general, the Uncertainty Principle is indeed a tool that gives informations on the possible localization of functions with respect to competing symmetries, that in this case are those of the well known group of translations and rotations of the Euclidean plane. The role of symmetries in the mechanisms of visual perception in V1 is a well recognized point \cite{CS, CF, KW}, as well as the Uncertainty Principle was already invoked to explain relevant cortical morphologies \cite{Daugman1985, BCSS}. Here we will use such concepts to characterize the resolution that can be obtained with joint spatial and angular measurements, based on the localization properties of receptive profiles. In terms of such characterizations we will deduce the bounds observed in Figure \ref{fig:Ring} as the result of intrinsic notions of balance between joint measurements resolutions.

\begin{figure}
\centering
\includegraphics[width=.6\textwidth]{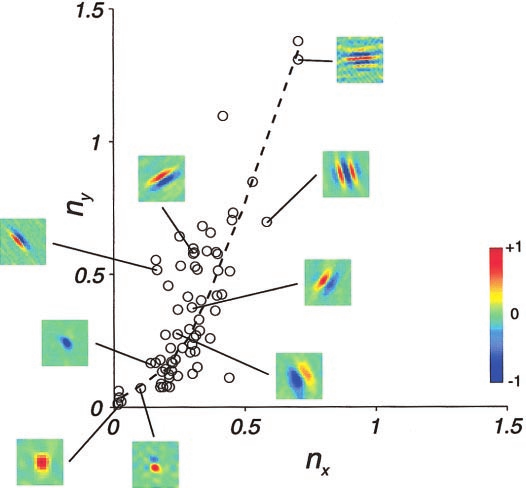}
\caption{Distribution of receptive profiles in terms of their shapes. Figure extracted from \cite{Ringach2002}.}\label{fig:Ring}
\end{figure}

\section{Receptive profiles and relevant symmetries}

We will assume isotropic gaussian Gabor filters as a model for standard V1 simple cells classical receptive profiles, defined on the Euclidean image plane:
\begin{equation}\label{eq:RP}
\psi_{q,p}^\sigma(x) = \frac{1}{\sigma \sqrt{\pi}} e^{ip\cdot (x-q)} e^{-\frac{|x-q|^2}{2\sigma^2}} \, , \qquad x \in \R^2
\end{equation}
with parameters $q \in \R^2$, $p \in \R^2$, $\sigma \in \R^+$. Each V1 simple cell is assumed to perform a linear filtering with a function shaped as in (\ref{eq:RP}), so that it can be characterized by these parameters. Their mapping on the two dimensional cortical layers are referred to as cortical maps \cite{BWF}. In particular, the centers $q$ of receptive fields are in a so-called retinotopic correspondence on the cortex \cite{HWFerrier}, while the size $\sigma$ is in average larger at the periphery and smaller close to the fovea \cite{HW}.  The frequency parameters $p$ are generally considered in polar coordinates $p = |p|(\cos\theta,\sin\theta)$, where $|p|$ is called spatial frequency and the angle $\theta$ up to a factor of $\pi$ is called preferred orientation, and their cortical maps are also well studied \cite{Ohki, KW, BCSS}.


The family of functions (\ref{eq:RP}) were proposed in \cite{Daugman1985} due to their optimal localization in space and frequency with respect to the classical Heisenberg Uncertainty Principle, and their fitness to model the linear behavior of simple cells was thoroughly tested \cite{Ringach2002}. We note however that here we are dealing with a simplified model of isotropic receptive fields, since as we will see this provides enough information for the present study, with the advantage that the results can be stated in a clearer form. We also recall that the real and imaginary parts in (\ref{eq:RP}) correspond to so-called even and odd cells
$$
\psi^\sigma_{q,p}(x) = \frac{1}{\sigma \sqrt{\pi}} \cos(p\cdot (x-q))e^{-\frac{|x-q|^2}{2\sigma^2}} + i \frac{1}{\sigma \sqrt{\pi}}\sin(p\cdot (x-q))e^{-\frac{|x-q|^2}{2\sigma^2}}
$$
but it will be sufficient for our purposes to deal with the full complex function as a whole.

\subsection{Groups of transformations}

Let us introduce the following unitary operators on $L^2(\R^2)$
\begin{itemize}
\item[i.] translations: $T_q f(x) = f(x - q)$, $q \in \R^2$
\item[ii.] modulations: $M_p f(x) = e^{ipx}f(x)$, $p \in \R^2$
\item[iii.] dilations: $\Sigma_\sigma f(x) = \frac{1}{\sigma}f(\frac{x}{\sigma})$, $\sigma \in \R^+$
\item[iv.] rotations: $R_\alpha f(x) = f(r_{-\alpha}x)$, $\alpha \in S^2$
\end{itemize}
where $r_{-\alpha}$ stands for the usual counterclockwise rotation of an angle $\theta$ on the Euclidean plane. In particular we note that rotations commute with dilations, and it is easy to see that
\begin{equation}\label{eq:RotMod}
R_\alpha M_p = M_{r_\alpha p} R_\alpha .
\end{equation}

If we denote with $g_1$ a $L^2(\R^2)$ normalized isotropic Gaussian with unit standard deviation
$$
g_1(x) = \frac{1}{\sqrt{\pi}} e^{-\frac{|x|^2}{2}}
$$
then we can characterize the functions (\ref{eq:RP}) in terms of the operators $i.$, $ii.$ and $iii.$ as
$$
\psi_{q,p}^\sigma(x) = T_q M_p \Sigma_\sigma g_1 (x) .
$$
Such a family is the prototype of a so-called wave packet systems \cite{CF1978}, and much is known about these structures \cite{KT, LWW}.

In this work we will deal with the localization properties of (\ref{eq:RP}) with respect to translations and local rotations, i.e. making use of the symmetries $i.$ and $iv.$, since they constitute two fundamental symmetries related to the mechanisms of visual perception in V1 (see e.g. \cite{CS} and references therein).

Local rotations are defined by
\begin{equation}\label{eq:localrotations}
R^q_\alpha f(x) \doteq f\left((r^q_\alpha)^{-1} x\right) = T_{q} R_{\alpha} T_{-q} f(x)
\end{equation}
where $r^q_\alpha x = r_\alpha(x - q) + q$ is a rotation of the Euclidean plane around point $q$, and with respect to these transformations we have the following.
\begin{lem}\label{lem:rotationofRP}
Let $\psi^\sigma_{q,p}$ be as in (\ref{eq:RP}). Then
\begin{equation}\label{eq:rotationofRP}
R^q_\alpha \psi^\sigma_{q,p}(x) = \psi^\sigma_{q,r_\alpha p}(x)
\end{equation}
\end{lem}
\begin{proof}
Using (\ref{eq:RotMod}) and the definition of local rotations (\ref{eq:localrotations}) we get
$$
R^q_\alpha \psi^\sigma_{q,p}(x) = T_{q} R_\alpha T_{-q} T_q M_p \Sigma_\sigma g_1(x) = T_{q} R_{\alpha} M_p \Sigma_\sigma g_1(x) = T_{q} M_{r_\theta p} R_{\alpha} \Sigma_\sigma g_1(x)
$$
so (\ref{eq:rotationofRP}) follows since rotations commute with dilations and $g_1$ is isotropic, i.e. $R_\alpha g_1(x) = g_1(x)$.
\end{proof}

Actually, the fact that $g_1$ is isotropic allows to write the whole family (\ref{eq:RP}) in terms of all the operators $i.$ to $iv.$ Indeed, denoting with $\theta$ the polar angle of $p$, that means $p = |p|(\cos\theta,\sin\theta)$, we can write (\ref{eq:RP}) as
$$
\psi_{q,p}^\sigma(x) = T_q R_\theta M_{\binom{|p|}{0}} \Sigma_\sigma g_1 (x)
$$
where $M_{\binom{|p|}{0}}f(x) = e^{i |p|x_1}f(x)$, so another way to characterize the system of functions (\ref{eq:RP}) is to consider a family $\{g_{\sigma,|p|} = M_{\binom{|p|}{0}} \Sigma_\sigma g_1 (x) \, , \ |p|,\sigma \in \R^+\}$ and rotate and translate each of its members. The aim of next section is actually to deduce properties on the localization of $\psi_{q,p}^\sigma$ with respect to the parameters $q$ and $\theta$, expressed in terms of the parameters $|p|$ and $\sigma$.

\section{Measures of Uncertainty}\label{sec:uncertainty}

In this section we characterize the uncertainty associated to joint measurements of positions and local orientations in terms of the properties of the measurement devices, expressed by $L^2$ functions, and quantify such uncertainties for the case of receptive profiles.

We recall that the generators $P_j$ of translations along cartesian axis are given by partial derivatives
\begin{equation}\label{eq:Mom}
P_j \doteq \frac{d}{dq_j}\Big|_{q = 0}T_q = \partial_{x_j}\ ,\quad j=1,2
\end{equation}
while the generator $\AV^q$ of a rotation around point $q$ can be written in terms of the ordinary infinitesimal rotation operator $\AV = x_2\partial_{x_1} - x_1\partial_{x_2}$
\begin{equation}\label{eq:AngMom}
\AV^q \doteq \frac{d}{d\alpha}\Big|_{\alpha = 0}R^q_\theta = T_{q} \AV T_{-q}
\end{equation}
and acts as the skew self-adjoint operator on $L^2(\R^2)$
$$
\AV^q f(x) = \frac{d}{d\alpha}\Big|_{\alpha = 0}f\left(r_\alpha^{-1}(x - q) + q\right) = \left((x_2 - q_2)\partial_{x_1} - (x_1 - q_1)\partial_{x_2}\right)f(x)\ .
$$

We will measure averages and variances using the standard definitions for operators on $L^2$, denoting with $\langle \cdot , \cdot \rangle_{L^2(\R^2)}$ the $L^2(\R^2)$ scalar product and with $\|\cdot\|_{L^2(\R^2)}$ the associated norm.
\begin{defi}\label{def:Averages}
Let $L$ be a densely defined skew self-adjoint linear operator on $\D \subset L^2(\R^2)$. We define its mean value over $f \in \D$ as
\begin{equation}\label{eq:average}
E_f(L) \doteq \langle (i L) f , f\rangle_{L^2(\R^2)} \ \in \R
\end{equation}
and its variance over $f \in \D$ as
\begin{equation}\label{eq:variance}
\left(\Delta_f L\right)^2 \doteq E_f\big((L - E_f(L))^2\big) = \|(L - E_f(L))\, f\|^2_{L^2(\R^2)} .
\end{equation}
\end{defi}
Since skew self-adjoint operators are the infinitesimal generators of a one parameter group of unitary transformations, the meaning of the average (\ref{eq:average}) is that of measuring the deformation of $f$ under such transformations
$$
E_f(L) = i \lim_{t \to 0} \frac{\langle (\exp{t L}) f , f\rangle_{L^2(\R^2)} - \langle f , f\rangle_{L^2(\R^2)}}{t}
$$
and the imaginary constant is merely a convention to ensure the result to be real. With this averaging, the variance (\ref{eq:variance}) has the usual meaning of strength of the fluctuations of $f$ under the considered transformations, that corresponds to the second moment of the distribution $t \mapsto \langle (\exp{t L}) f , f\rangle_{L^2(\R^2)}$. This means then that the variance (\ref{eq:variance}) provides a measure of the \emph{localization} of $f$ with respect to the symmetry $\exp{t L}$.

When applied to the operators (\ref{eq:Mom}) and (\ref{eq:AngMom}) of linear and rotational derivatives, these variances correspond respectively to a measure of linear and rotational fluctuations of a function $f$. The more $f$ is insensitive to translations ($f$ smooth and close to a constant function), the smaller is its $P$ variance, while a small $\AV^q$ variance means that $f$ has little sensitivity to rotations around $q$.

The notion of \emph{localization in orientation} that arises indicates that a function consisting of a set of parallel stripes, independently on their widths, is maximally localized in orientation, while a function that is circular symmetric around $q$ is minimally localized.

If we are interested in the joint localization properties of a function with respect to a two parameters group of unitary transformations, generated by two skew self-adjoint operators $L_1$ and $L_2$, we are led to consider the distribution
\begin{equation}\label{eq:distribution}
(t_1, t_2) \mapsto \langle (\exp{t_1 L_1})(\exp{t_2 L_2}) f , f\rangle_{L^2(\R^2)} .
\end{equation}
In this case, if the operators $L_1$ and $L_2$ do not commute, then the second moments of the distribution (\ref{eq:distribution}) are influenced by their commutator. Such an effect of competing symmetries is quantified by the Uncertainty Principle.

\subsection{The SE(2) Uncertainty Principle}

The operators (\ref{eq:Mom}) and (\ref{eq:AngMom}) satisfy the commutation relations of angular momentum \cite{CarruthersNieto1968}
\begin{equation}\label{eq:commut}
[\AV^q,P_1] = P_2 \ ; \quad [\AV^q,P_2] = - P_2\ .
\end{equation}
These commutators define the algebra of the $SE(2)$ group (see e.g. \cite{CS} and references therein), and for them the following generalized Uncertainty Principle holds \cite{FS, BCSS}, with respect to the quantities of Definition \ref{def:Averages}. Since we are dealing with densely defined operators, we will skip in what follows the technicalities related to operator domains, and refer the statements simply to $L^2(\R^2)$. For more details see \cite{FS}.

\begin{theo}[SE(2) Uncertainty Principle]
For any $f \in L^2(\R^2)$ it holds
\begin{equation}\label{eq:SE2uncertainties}
\left\{
\begin{array}{rcl}
\left(\Delta_f\AV^q\right)\, (\Delta_f P_1) & \geq & \frac12 \big|E_f(P_2)\big| \vspace{6pt}\\
\left(\Delta_f\AV^q\right)\, (\Delta_f P_2) & \geq & \frac12 \big|E_f(P_1)\big|
\end{array}
\right. .
\end{equation}
\end{theo}

These inequalities play the same role for the noncommutative symmetries of rotations and translations as the one played by the ordinary uncertainty inequality for the noncommutativity of quantum mechanical operators. The main difference is that in this case if we consider separately each of the two inequalities, we can not obtain a constant lower bound. Indeed for a function $f$ the product of variances of an infinitesimal rotations and a translations along one axis can be arbitrarily small, provided that the average of translations along the other axis on $f$ is small. This effect disappears when we consider translations on both axis, which is natural whenever we do not want to discriminate one direction over the other. In this case, we can actually recast the two inequalities (\ref{eq:SE2uncertainties}) into one inequality with a constant lower bound.

The following definition is closely related to that of \cite{Breitenberger1985}, and for this reason we use the same notation $\An$.

\begin{defi}\label{def:Angv}
Let us define the functional
$$
\An[f] \doteq \frac{\Delta_f P}{E_f(P)}
$$
where
$$
E_f(P) \doteq \sqrt{(E_f(P_1))^2 + (E_f(P_2))^2} \ \ \textnormal{and} \ \ \Delta_f P \doteq \sqrt{(\Delta_f P_1)^2 + (\Delta_f P_2)^2} .
$$
We denote with $\Delta \Theta[f]$ the corresponding measure of angular uncertainty
\begin{equation}\label{eq:DTheta}
\Delta \Theta [f] \doteq \arctan (\An[f]) .
\end{equation}
\end{defi}

With this definition, a direct consequence of the $SE(2)$ Uncertainty Principle is the following.
\begin{theo}
For all $f \in L^2(\R^2)$
\begin{equation}\label{eq:FULLuncertainty}
\left(\Delta_f\AV^q\right)\,\An[f] \geq \frac{1}{2} .
\end{equation}
\end{theo}

This inequality resembles the ordinary Heisenberg uncertainty inequality, since the presence of a constant lower bound provides a clear constraint on the joint localizations quantified by $\Delta_f\AV^q$ and $\An[f]$. However, as first noted in \cite{Jackiw1968}, the $SE(2)$ uncertainty inequalities (\ref{eq:SE2uncertainties}) can not be simultaneously minimized, so also (\ref{eq:FULLuncertainty}) does not admit minimizers. This is related to the issue of nonexistence of a canonically conjugate observable for angular momentum \cite{Dubin, Kastrup}. Indeed, if we had a well defined selfadjoint operator canonically commuting with angular momentum, we would end up with a well-known complex equation defining minimal uncertainty states \cite{FS}, while in this case we have two such equations, whose solutions provide CR function functions on the $\R^2 \times S^1$ for two noncompatible almost complex structures \cite{BC}.

\subsection{$SE(2)$ autocorrelations}

We pass now to the study of the properties of the distribution (\ref{eq:distribution}) applied to the symmetries under study, that we call autocorrelation since it has the form of the autocorrelation of a function with respect to the group of rotations and translations, and extends naturally the ordinary definition of autocorrelation with respect to translations. We will actually restrict the analysis to the square modulus of correlations, since as we will see it contains enough information for the present purposes. In particular we will show that such correlations can be used to characterize the uncertainty in the detection of position and local preferred angle associated to a function.

\begin{defi}
Given $f$ in $L^2(\R^2)$, we define its $SE(2)$ autocorrelation centered at $q$ as
\begin{equation}\label{eq:Corr}
\Corr_q[f](\xi,\alpha) = \big| \langle T_\xi R^q_\alpha f, f\rangle_{L^2(\R^2)}\big|^2 .
\end{equation}
\end{defi}

In general, $\Corr_q[f](\xi,\alpha)$ provides a natural way to study the joint localization properties of $f$ with respect to position and local preferred angle. Indeed, when we specialize to translations we get the usual autocorrelation, and by Plancherel theorem
\begin{equation}\label{eq:CorrTransl}
\Corr_q[f](\xi,0) = \big|\int_{\R^2}f(x -\xi)\overline{f(x)} dx \big|^2 = \big| \F(|\F f|^2)(\xi) \big|^2
\end{equation}
so we have that by Young inequality and Riemann-Lebesgue lemma $\Corr_q[f](\xi,0)$ is bounded and goes to $0$ as $\xi$ becomes large. Moreover, by the usual uncertainty principle we have that when $f$ is well localized in space, then $\F f$ is broadly localized, hence passing under another Fourier transform $\Corr_q[f](\xi,0)$ will decay rapidly, uniformly on $q$, and viceversa.

On the other hand, if we consider correlations only with respect to rotations, for simplicity centered at $q = 0$
\begin{equation}\label{eq:CorrRot}
\Corr_0[f](0,\alpha) = \big| \int_{\R^2}f(r_{-\alpha}x)\overline{f(x)} dx \big|^2
\end{equation}
essentially the same argument applies to the decay of correlations for functions that are localized with respect to rotations.

\begin{rem}[What does ``essentially the same argument'' mean]\ \\
Since $\F R_\theta f = R_\theta \F f$ we get
$$
\Corr_0[f](0,\alpha) = \big| \int_{\R^2}f(r_{-\alpha}x)\overline{f(x)} dx \big|^2 = \big| \int_{\R^2} R_\alpha (\F f)(k) \overline{\F f (k)} dk \big|^2
$$
so setting polar coordinates, with the notation $\phi_\kappa(\varphi) = (\F f)(\kappa\cos\varphi,\kappa\sin\varphi)$
$$
\Corr_0[f](0,\alpha) = \int_{\R^+} \kappa d\kappa \int_0^{2\pi} \phi_\kappa(\varphi - \alpha) \overline{\phi_\kappa(\varphi)} = \int_{\R^+} \kappa d\kappa \sum_{n \in \Z} e^{-2\pi i n\alpha} |\hat{\phi}_\kappa(n)|^2
$$
where $\hat{\phi}_\kappa(n) = \int_0^{2\pi} e^{-2\pi i n \varphi}\phi_\kappa(\varphi)d\varphi$ and the last transition is Parseval identity.

Since $L^2(\R^2) \approx L^2(\R^+) \otimes L^2(S^1)$ as tensor product of Hilbert spaces, and since $f$ is localized with respect to rotations in the real plane if and only if it is localized with respect to rotations in the Fourier plane, then we can assume without loss of generality that $\phi_\kappa(\varphi) = r(k) \Phi(\varphi)$, where $\Phi(\varphi)$ decays rapidly away from $\varphi = 0$ and $\int_{\R^+} |r(k)|^2 \kappa d\kappa = c < \infty$. So
$$
\Corr_0[f](0,\alpha) = c \sum_{n \in \Z} e^{-2\pi i n\alpha} |\hat{\Phi}(n)|^2
$$
and now strictly the same argument used for (\ref{eq:CorrTransl}) applies.
\end{rem}

\subsection{Uncertainty associated to $SE(2)$ measurements with receptive profiles}

When specialized to receptive profiles, the introduced uncertainties can be explicitly computed. In the proofs we will use the shorthand notation $g_\sigma = \Sigma_\sigma g_1$, and $\theta$ will be the polar angle of $p$.

\begin{lem}\label{lem:meanvarP}
The variance of the operators (\ref{eq:Mom}) on receptive profiles (\ref{eq:RP}) is
$$
\Delta_{\psi^\sigma_{qp}} P_j = \frac{1}{\sqrt{2}\sigma}\ .
$$
\end{lem}
\begin{proof}
Since $\partial_{x_j}\psi^\sigma_{qp}(x) = \left(i p_j - \frac{(x_j - q_j)}{\sigma^2}\right)\psi^\sigma_{qp}(x)$, we get
$E_{\psi^\sigma_{qp}}(P_j) = i p_j$, and
\begin{eqnarray*}
E_{\psi^\sigma_{qp}}(P_j^2) = \|P_j \psi^\sigma_{qp}\|^2_{L^2(\R^2)} & = & \int_{\R^2} \left|p_j + i \frac{(x_j - q_j)}{\sigma^2}\right|^2 |g_\sigma(x-q)|^2 dx \\
& = & p_j^2 + \frac{1}{\pi \sigma^2}\int_{\R^2} y_j^2 e^{- y^2} dy = p_j^2 + \frac{1}{2\sigma^2}\ .
\end{eqnarray*}
\end{proof}

\begin{lem}\label{lem:meanvarJ}
The variance of the operator (\ref{eq:AngMom}) on receptive profiles (\ref{eq:RP}) is
\begin{equation}\label{eq:AngMomVar}
\Delta \AV \doteq \Delta_{\psi^\sigma_{qp}} \AV^q = \frac{|p|\sigma}{\sqrt{2}}
\end{equation}
and we will call it angular momentum variance.
\end{lem}

\begin{proof}
Since
$$
\partial_{x_j}\psi^\sigma_{qp}(x) = (i p_j - \frac{(x_j - q_j)}{\sigma^2})\psi^\sigma_{qp}(x)
$$
then
$$
\AV^q \psi^\sigma_{qp}(x) = i \left((x_2 - q_2)p_1 - (x_1 - q_1)p_2\right)\psi^\sigma_{qp}(x)\ .
$$
Its mean value vanishes on $\psi^\sigma_{qp}$, due to the isotropy of $g_\sigma$:
\begin{eqnarray*}
E_{\psi^\sigma_{qp}}(\AV^q) & = & \langle \AV^q \psi^\sigma_{qp}, \psi^\sigma_{qp}\rangle_{L^2(\R^2)} = i \int_{\R^2} \left(x_2 p_1 - x_1 p_2\right) |g_\sigma(x)|^2 dx\\
& = & i |p| \int_{\R^2} \left(r_{-\theta}x\right)_2 |g_\sigma(x)|^2 dx = i |p| \int_{\R^2} x_2 |g_\sigma(x)|^2 dx = 0 .
\end{eqnarray*}
To compute the variance, by analogous arguments
\begin{eqnarray*}
(\Delta_{\psi^\sigma_{qp}} \AV^q)^2 & = & - |p|^2 \int_{\R^2} \left(\left(r_{-\theta}x\right)_2\right)^2 |g_\sigma(x)|^2 dx\\
& = & - \frac{|p|^2}{(\sigma \sqrt{\pi})^2} \left(\int_{\R} e^{-\frac{x_1^2}{\sigma^2}} dx_1\right) \left(\int_{\R} x_2^2 e^{-\frac{x_2^2}{\sigma^2}} dx_2\right)\\
& = & \frac{|p|^2 \sigma^2}{2\sqrt{\pi}} \int_{\R} y (-2 y e^{-y^2}) dy = \left(\frac{|p|\sigma}{\sqrt{2}}\right)^2\ .
\end{eqnarray*}
\end{proof}

We have then obtained the following proposition, which shows that for receptive profiles the angular momentum variance $\Delta\AV$ is inversely proportional to the angular uncertainty quantified in terms of $\An$.
\begin{prop}
Let $\psi^\sigma_{qp}$ be as in (\ref{eq:RP}). Then
$$
\Delta\AV\,\An[\psi^\sigma_{qp}] = \frac{1}{\sqrt{2}}\ .
$$
\end{prop}
\begin{proof}
Using Definition \ref{def:Angv}, Lemma \ref{lem:meanvarP} and Lemma \ref{lem:meanvarJ} we have
\begin{equation}\label{eq:AngVar}
(\An)[\psi^\sigma_{qp}] = \frac{1}{\sigma |p|} = \frac{1}{\sqrt{2}\Delta\AV}\ .
\end{equation}
\end{proof}

We will now consider $SE(2)$ autocorrelations of receptive profiles, and see that they indeed contain precisely the desired joint information on localizations in space and local orientation associated to the uncertainties we computed.
\begin{prop}\label{prop:correlations}
Let $\psi_{qp}^\sigma$ be defined by (\ref{eq:RP}). Then its $SE(2)$-autocorrelation reads
\begin{equation}\label{eq:CorrRP}
\Corr_q[\psi_{qp}^\sigma](\xi,\alpha) = e^{-\frac{|\xi|^2}{2\sigma^2}} e^{-(\Delta \AV)^2(1 - \cos\alpha)} .
\end{equation}
\end{prop}
\begin{proof}
By Lemma \ref{lem:rotationofRP}, and computing the Fourier transform of a gaussian
\begin{eqnarray*}
\langle T_\xi R^q_\alpha \psi^{\sigma}_{q,p} , \psi^{\sigma}_{q,p}\rangle_{L^2(\R^2)} & = & \langle T_\xi \psi^{\sigma}_{0,r_\alpha p} , \psi^{\sigma}_{0,p}\rangle_{L^2(\R^2)}\\
& = & \frac{1}{\sigma^2 \pi}\int_{\R^2} e^{i(r_\alpha p)(x - \xi)} e^{-\frac{|x - \xi|^2}{2\sigma^2}} e^{-ip x } e^{-\frac{|x|^2}{2\sigma^2}} dx\\
& = & \frac{e^{- i (r_{\alpha}p) \xi}}{\sigma^2 \pi}  e^{-\frac{|\xi|^2}{4\sigma^2}} \int_{\R^2} e^{-i(p - (r_\alpha p))x} e^{-\frac{|x - \xi/2|^2}{\sigma^2}} dx\\
& = & e^{-i \frac{p + (r_{\alpha}p)}{2}\xi} e^{-\frac{|\xi|^2}{4\sigma^2}} e^{-\frac{\sigma^2 |p - (r_\alpha p)|^2}{4}}
\end{eqnarray*}
so the result follows since $|p - (r_\alpha p)|^2 = 2|p|^2(1 - \cos\alpha)$.
\end{proof}

This proposition shows that the decay of the autocorrelation in space is a Gaussian with the same width of the corresponding receptive profile, that characterize spatial uncertainty. With respect to rotations, we have ended up with a Von Mises distribution in orientations. Such distributions appear naturally when discussing the the $SE(2)$ uncertainty principle \cite{CarruthersNieto1968, Hradil}, but they also provide a good model for orientation tuning of simple cells \cite{Swindale}, which is defined as the response curve of a cell to oriented stimuli \cite{Swindale, MBF}. This confirms that the introduced notion of localization is compatible with the resolution of measurements performed with receptive profiles. Moreover, we note that the commonly used circular variance \cite{JS} of the Von Mises distribution in (\ref{eq:CorrRP}), up to a normalization constant, is
$$
\textnormal{CircVar}(\Delta\AV) = 1 - \frac{I_1((\Delta\AV)^2)}{I_0((\Delta\AV)^2)}
$$
which results to be numerically close to what we have introduced as angular uncertainty (\ref{eq:DTheta}) when applied to receptive profiles (\ref{eq:AngVar})
$$
\left(\Delta\Theta[\psi^\sigma_{q,p}]\right)^2 = \left(\arctan\left(\frac{1}{\sqrt{2}\Delta\AV}\right)\right)^2 .
$$
In particular, as we will see in next section, typical values of $\Delta\AV$ in the filters encountered in V1 are around $1.7$, where the difference between these two notions of variance is around $5 \cdot 10^{-2}$.

\section{Bounds on the shape index induced by Uncertainty}

In this section we will use the measures of uncertainty referred to receptive profiles (\ref{eq:AngVar}) and (\ref{eq:CorrRP}) to deduce relevant features about the physiological data measured in \cite{Ringach2002} and depicted in Figure. In particular we will see how the informations provided by the analysis of uncertainty relations of Section \ref{sec:uncertainty} are sufficient to establish bounds on the number of subregions observed in the family of filters implemented in V1, and permit to reobtain characteristic sampling rates commonly used in image analysis.

A receptive profile $\psi^\sigma_{qp}$ consists of an oscillation of frequency $\nu = \frac{|p|}{2\pi}$ under a gaussian bell of width $\sigma$, so it appears natural to define a dimensionless index of shape \cite{Ringach2002}
\begin{equation}\label{eq:shape}
n = \nu \sigma .
\end{equation}
This quantity is related to the number of subregions defining a receptive profiles, since if we let $N_k$ be the number of half wavelength of receptive profile's oscillation within $k$ standard deviations $\sigma$, we obtain $N_k = 4 k n$. As it is apparent from the data measured in \cite{Ringach2002}, we see that approximately $k = 2$ standard deviations are sufficient to represent the main content of the filters, so that we can relate the effective subregions $N$ to $n$ as $N = 8n$,

In terms of $n$, the angular momentum variance (\ref{eq:AngMomVar}) of a receptive profile $\psi_{qp}$ reads
\begin{equation}\label{eq:AngMomVarn}
\Delta \AV = \sqrt{2} \pi n 
\end{equation}
while its angular variance (\ref{eq:DTheta}), after (\ref{eq:AngVar}), reads
\begin{equation}\label{eq:DThetan}
\Delta \Theta [\psi^\sigma_{qp}]= \arctan\left(\frac{1}{2 \pi n}\right)\ .
\end{equation}

\subsection{Lower bound for orientation measurements}

By the discussions in Section \ref{sec:uncertainty}, we have seen how we can quantify with $\Delta\Theta$ the angle resolution allowed by a linear filtering. If we refer to the task of orientation detection, we can set as a reasonable bound that of angle uncertainty less than $\pi/2$, that is expressed by
$$
\Delta \Theta [\psi^\sigma_{qp}] \leq \frac{\pi}{4} .
$$
This condition can be stated in terms of the shape index using (\ref{eq:DThetan})
$$
n \geq \frac{1}{2\pi} \approx 0.16 \doteq n_\m .
$$
As we can see in Figure and by the discussions in \cite{Ringach2002}, cells which show a selectivity in orientation all lie above this threshold. Moreover, we note that for indexes $n < n_\m$ it can be a hard task to distinguish an even cell from being represented only by a gaussian, while odd cells under this threshold all appear identical up to a multiplicative factor, so the parametric fit of the Gabor model (\ref{eq:RP}) is quite delicate in this region. We can then interpret the bunch of broadly tuned cells around the zero value of the shape index $n$ as generally below the minimal uncertainty bound that allows a consistent detection of orientations.

\subsection{Upper bound}

In order to discuss the upper bound, we introduce a notion of characteristic length associated to a specific level set of the correlations (\ref{eq:RP}), intrinsically related to the task of detection of positions and local orientations. Its purpose is to quantify the minimum distance that one needs to cover in order to decorrelate a function $f$ as much as $f$ is decorrelated when compared at orthogonal directions.


\begin{defi}
The correlation length for $f \in L^2(\R^2)$ is the smallest distance $\lambda$ for which
\begin{equation}\label{eq:separationcondition}
\Corr[f](\xi,0) \leq \Corr[f](0,\frac{\pi}{2}) \qquad \forall\, |\xi| = \lambda\, .
\end{equation}
\end{defi}

If we apply this notion to receptive profiles (\ref{eq:RP}) we obtain the following.
\begin{prop}
The shape index (\ref{eq:shape}) is bounded from above by the ratio of the correlation length $\lambda$ and the spatial uncertainty $\sigma$
\begin{equation}\label{eq:separationconditionrevisited}
2 \pi n \leq \frac{\lambda}{\sigma} \, .
\end{equation}
\end{prop}
\begin{proof}
Condition (\ref{eq:separationcondition}) on receptive profiles $\psi_{qp}^\sigma$, by (\ref{eq:CorrRP}) reads
$\Delta\AV \leq \frac{\sqrt{2}}{2}\frac{\lambda}{\sigma}$, since
$$
e^{-\frac{\lambda^2}{2\sigma^2}} \leq e^{-\Delta\AV^2} \ \iff \ \frac{\lambda^2}{2\sigma^2} \geq \Delta\AV^2
$$
so (\ref{eq:separationconditionrevisited}) follows by the relation (\ref{eq:AngMomVarn}) between $\Delta\AV$ and the shape index.
\end{proof}

On the other hand, as discussed when dealing with the relation between the shape index and the number of subregions, we have also that the effective field of influence of a receptive profile can be set within two standard deviations $\sigma$.
From this point of view, we can then assume that the distance $d$ at which a receptive profile is effectively spatially uncorrelated corresponds to the distance that one has to cover in order to let its effective effective field of influence not intersect with its translation at a distance $d$, i.e. $d = 4\sigma$.

In order to couple with both position and orientation measurements, we will then consider the hypothesis of balance of the two characteristic scales introduced, that is the identification $\lambda = d$. By (\ref{eq:separationconditionrevisited}), this condition can be stated in terms of the shape index as
$$
\quad n \leq \frac{2}{\pi} \approx 0.64 \doteq n_\M .
$$
To compare this bound with Figure, we recall that here we are dealing with the simplified model of isotropic receptive fields, while in \cite{Ringach2002} the analysis is performed considering two anisotropic indexes $n_x$ and $n_y$. In terms of such indexes we can see that the largest part of the population lies within two bounds $n_x \lesssim 0.5$ and $n_y \lesssim 0.76$, and $n_\M$ looks in good accordance with their mean value.

The question of whether this identification of characteristic distances is truly implemented in the cortex cannot be answered at this point, but we note that a cortical scale related to the symmetries under study that is possibly compatible with the proposed relation is the mean correlation length of orientation preference maps (see e.g. \cite{BWF, BCSS} and references therein). Indeed, by the measurements performed in \cite{Fitzpatrick2002} we see that such scale is comparable with the size of a so-called cortical point image, that is the cortical region that is activated after a highly spatially localized stimulus, and at least when we reduce to linear behavior of cells this notion corresponds to what we have indicated as effective field of influence.


\subsection{Sampling on orientations}
Another intriguing consequence of the performed uncertainty analysis can be stated in terms of optimal sampling rates for orientation detection. Indeed if we consider the mean value on shape index measured in \cite{Ringach2002}, or equivalently, in terms of the deduced bounds, for $n = \frac{n_\M + n_\m}{2} \approx 0.4$, we have that
$$
\Delta \Theta [\psi^\sigma_{qp}] = \arctan\left(\frac{1}{0.8 \pi}\right) \approx \frac{\pi}{8} .
$$
With respect to Gabor filters possessing such $n$, one way to use such result is to consider that the detection of orientations at angles that are closer than this uncertainty do not provide an actual improvement in the resolution of the local orientation present in the stimulus, so that it can be sufficient to cover the interval of orientations $[0,\pi)$ with a sampling having a $\pi/8$ spacing. This actually compares well with the notions of optimal sampling adopted in image analysis tasks (see e.g. in \cite{Lee} and references therein), generally justified with independent arguments. Moreover, this uncertainty analysis permits to set clear sampling spacings depending on the shape index of the filter used.

\section{Conclusions}

In this paper we have studied theoretical aspects of an analytic characterization of uncertainty that generalizes the well known Heisenberg Uncertainty Principle to the symmetries associated with the task of joint measurements of position and local orientation. The implications of this analysis, together with an hypothesis of balance between characteristic correlation distances, allowed us to obtain bounds comparable with experimental data on the shape index of the V1 simple cells that are selective for orientation, and to separate them from broadly tuned cells, which lie below the uncertainty bound for consistent orientation detection.

We remark that this was possible even if our working assumptions on the functional behavior of simple cells were reduced to linear filtering with symmetric receptive fields, and the only considered task is the one associated to the sole symmetries of rotations and translations.

Whether such elementary principles could be directly responsible of the observed distribution of receptive profiles is a question that can hardly find an answer. Nevertheless, the present study shows that they are sufficient to describe many of the relevant features that concern the shape of simple cells.

\

\

\noindent
\textsc{Davide Barbieri}: CAMS, EHESS, Paris. \emph{davide.barbieri8@gmail.com}\\
\textsc{Giovanna Citti}: Dept. Mathematics, Bologna.
\emph{giovanna.citti@unibo.it}\\
\textsc{Alessandro Sarti}: CAMS, EHESS, Paris. \emph{alessandro.sarti@ehess.fr}


\begin{thebibliography}{99}
\bibitem{BC} D. Barbieri, G. Citti, \emph{Coherent states of the Euclidean Motion Group and CR regularity}. Submitted.  http://arxiv.org/abs/1301.3783
\bibitem{BCSS} D. Barbieri, G. Citti, G. Sanguinetti, A. Sarti, \emph{An uncertainty principle underlying the functional architecture of V1}. J. Physiol. Paris 106(5-6):183-193 (2012).
\bibitem{BWF} A. Basole, L. E. White, D. Fitzpatrick, \emph{Mapping multiple features in the population response of visual cortex}. Nature 423:986-990 (2003).
\bibitem{Breitenberger1985} E. Breitenberger,  \emph{Uncertainty measures and uncertainty relations for angle observables}. Found. Phys. 15:353-364 (1985).
\bibitem{CF} P. Chossat, O. Faugeras, \emph{Hyperbolic planforms in relation to visual edges and textures perception}. PLoS Comput. Biol. 5(12):1-16 (2009).
\bibitem{Fitzpatrick2002} W. H. Bosking, J. C. Crowley, D. Fitzpatrick, \emph{Spatial coding of position and orientation in primary visual cortex}. Nature Neurosc. 5:874-882 (2009).
\bibitem{CarruthersNieto1968} P. Carruthers, M. M. Nieto, \emph{Phase and angle variables in quantum mechanics}. Rev. Mod. Phys. 40:441-440 (1968).
\bibitem{CS} G. Citti, A. Sarti, \emph{A cortical based model of perceptual completion in the roto-translation space}. J. Math. Imag. Vis. 24(3):307-326 (2006)
\bibitem{CF1978} A. C\'{o}rdoba, C. Fefferman, \emph{Wave packets and Fourier integral operators}. Comm. Partial Diff. Eq. 3:979-1005 (1978).
\bibitem{Daugman1985} J. G. Daugman, \emph{Uncertainty relation for resolution in space, spatial frequency, and orientation optimized by two dimensional visual cortical filters}. J. Opt. Soc. Am. A 2:1160-1169 (1985).
\bibitem{Dubin} D. A. Dubin, M. A. Hennings, T. B. Smith, \emph{Quantization in polar coordinates and the phase operator}. Publ. RIMS Kyoto 30:479-532 (1994).
\bibitem{FS} G. B. Folland, A. Sitaram, \emph{The uncertainty principle: a mathematical survey}. J. Fourier Anal. Appl. 3:207-238. (1997)
\bibitem{G} N. V. Graham, \emph{Beyond multiple pattern analyzers modeled as linear filters (as classical V1
simple cells): Useful additions of the last 25 years}. Vision Research 51:1397–1430 (2011).
\bibitem{H} H. Hosoya, \emph{Multinomial bayesian learning for modeling classical and nonclassical receptive field properties}. Neural Comput. 24(8):2119-2150 (2012).
\bibitem{Hradil} Z. Hradil, J. \v{R}eh\'{a}\u{c}ek, Z. Bouchal, R. \u{C}elechovsk\'{y}, L. L. S\'{a}nchez-Coto, \emph{Minimum uncertainty measurement of angle and angular momentum}. Phys. Rev. Lett. 97:243601 (2006).
\bibitem{HW} D.H. Hubel, T.N. Wiesel, \emph{Uniformity of monkey striate cortex: a parallel relationship between field size, scatter, and magnification factor}. J. Comp. Neurol. 158(3):295-305 (1974).
\bibitem{HWFerrier} D.H. Hubel, T.N. Wiesel, \emph{Functional architecture of macaque monkey visual cortex}. Proc. R. Soc. Lond. B. 198:1-59 (1977).
\bibitem{Jackiw1968} R. Jackiw, \emph{Minimum Uncertainty Product, Number Phase Uncertainty Product and Coherent States}. J. Math. Phys. 9:339 (1968).
\bibitem{JS} S. R. Jammalamadaka, A. Sengupta, \emph{Topics in circular statistics}. World Scientific (2001).
\bibitem{JP1} J. P. Jones, L. A. Palmer, \emph{An evaluation of the two-dimensional Gabor filter model of simple receptive fields in cat striate cortex}. J Neurophysiol 58:1233–1258 (1987).
\bibitem{JP2} J. P. Jones, L. A. Palmer, \emph{The two-dimensional spatial structure of simple receptive fields in cat striate cortex}. J Neurophysiol 58:1187–1211 (1987).
\bibitem{KT} C. Kalisa, B. Torr\'esani, \emph{N-dimensional affine Weyl-Heisenberg wavelets}. Ann. Inst. H. Poincar\'e Phys. Th\'eor., 59(2):201–236 (1993).
\bibitem{Kastrup} H. A. Kastrup, \emph{Quantization of the canonically conjugate pair angle and orbital angular momentum}. Phys. Rev. A 73:052104 (2006).
\bibitem{KW} W. Keil, F. Wolf, \emph{Coverage, continuity and visual cortical architecture}. Neural Systems and Circuits 1(17):doi:10.1186/2042-1001-1-17 (2011).
\bibitem{LWW} D. Labate, G. Weiss, E. Wilson, \emph{An approach to the study of wave packet systems}. Contemp. Math., Wavelets, Frames Operator Theory 345:215–235 (2004).
\bibitem{Lee} T. S. Lee, \emph{Image representation using 2D Gabor Wavelets}. IEEE Transactions on Pattern Analysis and Machine Intelligence 18:959-971(1996).
\bibitem{MBF} F. Mooser, W. H. Bosking, D. Fitzpatrick, \emph{A morphological basis for orientation tuning in primary visual cortex}. Nature Neurosc. 7:872-879 (2004).
\bibitem{Ohki} K. Ohki, S. Chung, P. Kara1, M. Hubener, T. Bonhoeffer, R. C. Reid, \emph{Highly ordered arrangement of single neurons in
orientation pinwheels}. Nature 442:925-928 (2006).
\bibitem{OlshausenField} B. A. Olshausen, D. J., Field, \emph{Sparse coding with an overcomplete basis set: a strategy employed by V1?}. Vision Res 37:3311–3325 (1997).
\bibitem{Ringach2002} D. L. Ringach, \emph{Spatial structure and symmetry of simple cell receptive fields in Macaque primary visual cortex}. J. Neurophysiol. 88:455-463 (2002).
\bibitem{SSBSL} J. Shelton, P. Sterne, J. Bornschein, A. S. Sheikh, J. Luecke,  \emph{Why MCA? Nonlinear sparse coding with spike-and-slab prior for neurally plausible image encoding}. Adv. Neural Inf. Proc. Systems 25:2285-2293 (2012).
\bibitem{Swindale} N. V. Swindale, \emph{Orientation tuning curves: empirical description and estimation of parameters}. Biol. Cybern. 78:45–56 (1998)
\end{thebibliography}
\end{document}